\documentclass[11pt,letter]{article}
\usepackage{authblk}
\usepackage[margin=1in]{geometry}
\usepackage[utf8x]{inputenc}
\usepackage{amsthm}
\usepackage{wrapfig,floatflt,graphicx,amssymb,textcomp,array,amsmath}
\usepackage{enumerate,enumitem}
\usepackage{multirow}
\usepackage{tabularx}
\usepackage{color,xcolor}

\definecolor{mycolor}{rgb}{0, 0, 0}
\usepackage{todonotes}
\usepackage[titletoc,title]{appendix}

\usepackage{lineno}
\newcommand{\etal}{{et~al.}}
\newcommand{\conv}[1]{\mathrm{CH}(#1)}

\title{Improved Bounds for Covering Paths and Trees in the Plane}

\author{Ahmad Biniaz\thanks{School of Computer Science, University of Windsor, \texttt{abiniaz@uwindsor.ca}. Research supported by NSERC.}}

\date{}
\newtheorem{lemma}{Lemma}
\newtheorem{corollary}{Corollary}

\newtheorem{theorem}{Theorem}

\newtheorem*{problem*}{Problem}
\newtheorem*{claim*}{Claim}
\newtheorem*{invariant*}{Invariant}

\begin{document}
	\maketitle
	\begin{abstract}
A {\em covering path} for a planar point set is a path drawn in the plane with straight-line edges such that every point lies at a vertex or on an edge of the path. A {\em covering tree} is defined analogously. Let $\pi(n)$  be the minimum number such that every set of $n$ points in the plane can be covered by a noncrossing path with at most $\pi(n)$ edges. Let $\tau(n)$ be the analogous number for noncrossing covering trees. Dumitrescu, Gerbner, Keszegh, and
T{\'{o}}th (Discrete \& Computational Geometry, 2014) established the following inequalities:

\[\frac{5n}{9} - O(1)  <\pi(n)  < \left(1-\frac{1}{601080391}\right)n, \text{ ~~~and ~~~}\frac{9n}{17} - O(1)<\tau(n)\leqslant \left\lfloor\frac{5n}{6}\right\rfloor.\] 
We report the following improved upper bounds:
 \[\pi(n)\leqslant \left(1-\frac{1}{22}\right)n, \text{ ~~~and ~~~} \tau(n)\leqslant \left\lceil\frac{4n}{5}\right\rceil. \]

In the same context we study rainbow polygons. For a set of colored points in the plane, a {\em perfect rainbow polygon} is a simple polygon
that contains exactly one point of each color in its interior or on its boundary.
Let $\rho(k)$ be the minimum number such that every $k$-colored~point set in the plane admits a perfect rainbow polygon of size $\rho(k)$. Flores{-}Pe{\~{n}}aloza, Kano, Mart{\'{\i}}nez{-}Sandoval, Orden, Tejel, T{\'{o}}th, Urrutia, and
Vogtenhuber (Discrete Mathematics, 2021) proved that $20k/19 - O(1)  <\rho(k)  < 10k/7 + O(1).$
We report the improved upper bound $\rho(k)< 7k/5 + O(1)$. 

To obtain the improved bounds we present simple $O(n\log n)$-time algorithms that achieve paths, trees, and polygons with our desired number of edges. 

%The improved bounds are obtained by simple algorithms that run in $O(n\log n)$ time.
\end{abstract}

\section{Introduction}
\begin{wrapfigure}{r}{0.18\textwidth}
	\begin{center}
		\vspace{-20pt}
		\includegraphics[width=.15\textwidth]{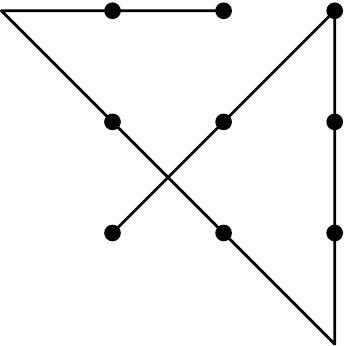}
	\end{center}
	\vspace{-20pt}
\end{wrapfigure}

Traversing a set of points in the plane by a polygonal path possessing some desired properties has a rich background. For example the famous traveling salesperson path problem asks for a polygonal path with minimum total edge length \cite{Arora1998,Papadimitriou1977}. In recent years there has been an increased interest in paths with properties such as being noncrossing \cite{Aichholzer2010,Cerny2007}, minimizing the longest edge length \cite{Biniaz2022}, maximizing the shortest edge length \cite{Arkin1999}, minimizing the total or the largest turning angle \cite{Aggarwal1999, Fekete1997}, and minimizing the number of turns (which is the same as minimizing the number of edges) \cite{Dumitrescu2014,Stein2001} to name a few.

The main focus of this paper is polygonal paths with a small number of edges. It is related to the classical nine dots puzzle which asks for covering the vertices of a $3\!\times\! 3$ grid by a polygonal path with no more than 4 segments. It appears in Sam Loyd's  Cyclopedia of Puzzles from 1914 \cite{Loyd1914}.

Let $P$ be a set of $n$ points in the plane. A {\em spanning path} for $P$ is a path drawn in the plane with straight-line edges such that every point of $P$ lies at a vertex of the path and every vertex of the path lies at a point of $P$. In other words, it is a Hamiltonian path which has exactly $n-1$ edges. The path in the figure above is not a spanning path because two of its vertices do not lie on given points. A {\em covering path} for $P$ is a path drawn in the plane with straight-line edges such that every point of $P$ lies at a vertex or on an edge of the path. A vertex of a covering path can be any point in the plane (not necessarily in $P$). The path in the figure above is a covering path with $4$ edges. With these definitions, any spanning path is also a covering path, but a covering path may not be a spanning path. A {\em covering tree} for $P$ is defined analogously as a tree drawn in the plane with straight-line edges such that every point of $P$ lies at a vertex or on an edge of the tree. A covering path or a tree is called {\em noncrossing} if its edges do not cross each other. The edges of covering paths and trees are also referred to as {\em links} in the literature \cite{Arkin2003b}. 

Covering paths and trees have received considerable attention in recent years, see e.g. \cite{Arkin2003b,Dumitrescu2014,Keszegh2013}. In particular covering paths with a small number of edges find applications in robotics and heavy machinery for which turning is an expensive operation \cite{Stein2001}. Covering trees with a small number of edges are useful in red-blue separation \cite{Fulek2013} and in constructing rainbow polygons \cite{Flores-Penaloza21}. In 2010 F.~Mori\'{c} \cite{Demaine2011} and later Dumitrescu, Gerbner, Keszegh, and
T{\'{o}}th \cite{Dumitrescu2014} raised many challenging questions about covering paths and trees. Specifically they asked the following two questions which are the main topics of this paper. As noted in \cite{Demaine2011}, analogous questions were asked by E.~Welzl in Gremo’s Workshop on Open Problems 2011.

\begin{enumerate}
	\item What is the minimum number $\pi(n)$ such that every set of $n$ points in the plane can be covered by a noncrossing path with at most $\pi(n)$ edges?
	\item What is the minimum number $\tau(n)$ such that every set of $n$ points in the plane can be covered by a noncrossing tree with at most $\tau(n)$ edges?
\end{enumerate}

For both $\pi(n)$ and $\tau(n)$, a trivial upper bound is $n-1$ (which comes from the existence of a noncrossing spanning path) and a trivial lower bound is $\lceil\frac{n}{2}\rceil$ (because if no three points are collinear then each edge covers at most two points). In 2014,  Dumitrescu \etal~\cite{Dumitrescu2014} established, among other interesting results, the following nontrivial bounds: 
\[\frac{5n}{9} - O(1)  <\pi(n)  < \left(1-\frac{1}{601080391}\right)n, \text{ ~~~and ~~~}\frac{9n}{17} - O(1)<\tau(n)\leqslant \left\lfloor\frac{5n}{6}\right\rfloor.\]

The following is a related question that has recently been raised by Flores{-}Pe{\~{n}}aloza, Kano, Mart{\'{\i}}nez{-}Sandoval, Orden, Tejel, T{\'{o}}th, Urrutia, and
Vogtenhuber \cite{Flores-Penaloza21} in the context of rainbow polygons. For a set of colored points in the plane, a {\em rainbow polygon} is a simple polygon
that contains at most one point of each color in its interior or on its boundary. A rainbow polygon is called  {\em perfect} if it contains exactly one point of each color. The {\em size} of a polygon is the number of its edges (which is the same as the number of its vertices).

\begin{enumerate}
	\item[3.] What is the minimum number $\rho(k)$ (known as {\em the rainbow index}) such that every $k$-colored point set in the plane, with no three collinear points, admits a perfect rainbow polygon of size $\rho(k)$?
\end{enumerate}

Question 3 is related to covering trees in the sense that (as we will see later in Section~\ref{polygon-section}) particular covering trees could lead to better upper bounds for $\rho(k)$. Flores{-}Pe{\~{n}}aloza \etal~\cite{Flores-Penaloza21} established the following inequalities: $$\frac{20k}{19} - O(1)  <\rho(k)  < \frac{10k}{7} + O(1).$$

The upper bounds on $\pi(n)$, $\tau(n)$, and $\rho(n)$ are universal (i.e., any point set admits these bounds) and they are obtained by algorithms that achieve paths, trees, and polygons of certain size \cite{Dumitrescu2014,Flores-Penaloza21}. The lower bounds, however, are existential (i.e., there exist point sets that achieve these bounds) and they are obtained by the same point set that is exhibited in \cite{Dumitrescu2014}. Perhaps there should be configurations of points that achieve better lower bounds for each specific number. 

\subsection{Our Contributions}
%The known lower bounds for the three questions are obtained by a same point set which is exhibited in \cite{Dumitrescu2014}. 

Narrowing the gaps between the lower and upper bounds for $\pi(n)$, $\tau(n)$, and $\rho(n)$ are open problems which are explicitly mentioned in \cite{Dumitrescu2014,Flores-Penaloza21}.
In this paper we report the following improved upper bounds for the three numbers:  
\[\pi(n)\leqslant \left(1-\frac{1}{22}\right)n, \text{ ~~~~} \tau(n)\leqslant \left\lceil\frac{4n}{5}\right\rceil, \text{ ~~~and ~~~}\rho(k)< \frac{7k}{5} + O(1). \]
The new bounds for $\pi(n)$ and $\tau(n)$ are the first improvements in 8 years. To obtain these bounds we present algorithms that achieve noncrossing covering paths, noncrossing covering trees, and rainbow polygons with our desired number of edges. The algorithms are simple and run in $O(n\log n)$ time where $n$ is the number of input points. 
The running time is optimal for paths because computing a noncrossing covering path has an $\Omega(n \log n)$ lower bound \cite{Dumitrescu2014}.  A noncrossing covering tree, however, can be computed in $O(n)$ time by taking a spanning star. We extend our path algorithm and achieve an upper bound of $(1-\frac{1}{22})n+2$ for noncrossing {\em covering cycles}. This is a natural variant of the traveling salesperson tour problem with the objective of minimizing the number of links, which is NP-hard \cite{Arkin2003b}.

Our algorithms share some similarities with previous algorithms in the sense that both are iterative and use the standard plane sweep technique which scans the points from left to right. However, to achieve the new bounds we employ new geometric insights and make use of convex layers and the Erd\H{o}s-Szekeres theorem \cite{Erdos1935}.

Regardless of algorithmic implications, our results are important because they provide new information on universal numbers $\pi(n)$, $\tau(n)$, and $\rho(n)$ similar to the crossing numbers \cite{Aichholzer2020,Czabarka2008,Harary1963}, the size of crossing families (pairwise crossing edges) \cite{Pach2019}, the Steiner ratio \cite{Arora1998,Ivanov2012}, and other numbers and constants studied in discrete geometry (such as \cite{Biniaz2022,Chan2004,Fekete2000}).

%\vspace{8pt}
%\noindent{\em An assumption.} Collinear points are beneficial for covering paths and trees as they usually lead to paths and trees with fewer edges. In our algorithms (which consider constant number of points in each iteration), collinear points could be simply handled by considering more cases. Therefore, to avoid the interruption of our arguments we assume that no three points are collinear.

\vspace{8pt}
{\color{mycolor}\noindent{\em Remark.} Collinear points are beneficial for covering paths and trees as they usually lead to paths and trees with fewer edges. To avoid the interruption of our arguments we first describe our algorithms for point sets with no three collinear points. In the end we show how to handle collinearities.}

\subsection{Related Problems and Results}
If we drop the noncrossing property, Dumitrescu \etal~\cite{Dumitrescu2014} showed that every set of $n$ points in the plane admits a (possibly self-crossing) covering path with $n/2 + O(n/\log n)$ edges. 
Covering paths have also been studied from the optimization point of view. The problem of computing a covering path with minimum number of edges for a set of points in the plane (also known as the {\em minimum-link} covering path problem and the {\em minimum-bend} covering path problem) is shown to be NP-hard by Arkin \etal~\cite{Arkin2003b}.  Stein and Wagner \cite{Stein2001} presented an $O(\log z)$-approximation algorithm where $z$ is the
maximum number of collinear points.

Keszegh \cite{Keszegh2013} determined exact values of $\pi(n)$ and $\tau(n)$ for vertices of square grids. The axis-aligned version of covering paths is also well-studied and various lower bounds, upper bounds, and approximation algorithms are presented to minimize the number of edges of such paths; see e.g.  \cite{Bereg2009,Collins2004,Jiang2015}. Covering trees are studied also in the context of separating red and blue points in the plane \cite{Fulek2013}.
The problem of covering points in the plane with minimum number of lines is another related problem which is also well-studied, see e.g. \cite{Chen2020,Grantson2006,Langerman2005}.

For problems and results related to rainbow polygons we refer the reader to the paper of Flores{-}Pe{\~{n}}aloza \etal~\cite{Flores-Penaloza21}. In particular, they determine the exact rainbow indices for small values of $k$ by showing that $\rho(k)=k$ for $k\in\{3,4,5,6\}$ and $\rho(7)=8$.

\subsection{Preliminaries}
For two points $p$ and $q$ in the plane we denote by $\ell(p,q)$ the line through $p$ and $q$, and by $pq$ the line segment with endpoints $p$ and $q$.
For two paths $\delta_1$ and $\delta_2$, where $\delta_1$ ends at the same vertex at which $\delta_2$ starts, we denote their concatenation by $\delta_1 \oplus \delta_2$. 

A point set $P$ is said to be in {\em general position}
if no three points of $P$ are collinear. We denote the convex hull of $P$ by $\conv{P}$. A set $K$ of $k$ points in the plane in convex position, with no two points on a vertical line, is
a $k$-{\em cap} (resp. a $k$-{\em cup}) if all points of
$K$ lie on or above (resp. below) the line through the leftmost and rightmost
points of $K$.
A classical result of Erd\H{o}s and Szekeres \cite{Erdos1935} implies that every set of at least ${{2k-4}\choose{k-2}} + 1$ points in the plane in general position, with no two points on a vertical line, contains
a $k$-cap or a $k$-cup. This bound is tight in the sense that there are point sets of size
${{2k-4}\choose{k-2}}$ that do not contain any $k$-cap or $k$-cup \cite{Erdos1935}.

% {\color{mycolor}As noted in ~\cite{Dumitrescu2014}, the notion of cap and cup can be extended to arbitrary point sets in the plane (with collinearities) in which case $K$ is in {\em weakly} convex position. Also, the result of \cite{Erdos1935} gives the same quantitative bounds for caps or cups in weakly convex position.}

\section{Noncrossing Covering Paths}
\label{path-section}

In this section we prove that $\pi(n)\leqslant (1-1/22)n$.
We start by the following folklore result on the existence of noncrossing polygonal paths among points in the plane; see e.g. \cite{Dumitrescu2014,Fulek2013}.

\begin{lemma}
	\label{convex-lemma}
	Let $P$ be a set of 
	%$n$ 
	points in the plane in the interior of a convex region C, and let $p$ and $q$ be two points on the boundary of C. Then $P\cup\{p,q\}$ admits a noncrossing spanning path with $|P|+1$ edges such that
	% its internal vertices are the points of $P$, 
	its endpoints are $p$ and $q$, and its relative interior lies in the interior of $C$. 
	%Moreover, such a path can be found in $O(n\log n)$.
\end{lemma} 

In fact the spanning path that is obtained by Lemma~\ref{convex-lemma} is a noncrossing covering path for $P\cup \{p,q\}$ and it lies in the convex hull of $P\cup \{p,q\}$. The following lemma shows that any set of $23$ points can be
covered by a noncrossing path with $21$ edges.

\begin{lemma}
	\label{27-lemma}
	{\color{mycolor}Let $P$ be a set of at least $23$ points in the plane such that no two points have the same $x$-coordinate. Let $H$ be the vertical strip  bounded by the vertical lines through the leftmost and rightmost points of $P$. Then there exists a noncrossing covering path for $P$ with $|P|-2$ edges that is contained in $H$ and its endpoints are the leftmost and rightmost points of $P$.}
\end{lemma}

\begin{proof}
Our proof is constructive. 	  Let $l$ and $r$ be the leftmost and rightmost points of $P$, respectively. 		
	Let $P'=P\setminus\{l,r\}$, and notice that $|P'|\geqslant 21$. 
	{\color{mycolor} We assume that $P'$ is in general position. In the end of the proof we briefly describe how to handle collinearities.}
	By the result of \cite{Erdos1935} the set $P'$ has a $5$-cap or a $5$-cup. After a suitable reflection we may assume that it has a $5$-cup $K$ with points $p_1,p_2,p_3,p_4,p_5$ from left to right, as in Figure~\ref{path-fig1}(a). 
	Among all $5$-cups in $P'$ we may assume that $K$ is one for which $p_1$ is the leftmost possible point. Also among all such $5$-cups (with leftmost point $p_1$) we may assume that $K$ is the one for which $p_5$ is the rightmost possible point. This choice of $K$ implies that the region that is the intersection of $H$ with the halfplane above $\ell(p_1,p_2)$ and the halfplane to the left of the vertical line through $p_1$ is empty of points of $P'$; this region is denoted by $E_1$ in Figure~\ref{path-fig1}(a). Similarly the region that is the intersection of $H$ with the halfplane above $\ell(p_4,p_5)$ and the halfplane to the right of the vertical line through $p_5$ is empty of points of $P'$; this region is denoted by $E_2$ in Figure~\ref{path-fig1}(a).

	For brevity let $\ell_{12}=\ell(p_1,p_2)$ and $\ell_{45}=\ell(p_4,p_5)$. 
	We distinguish two cases: (i) $l$ lies below $\ell_{12}$ or $r$ lies below $\ell_{45}$, and (ii) $l$ lies above $\ell_{12}$ and $r$ lies above $\ell_{45}$.

\vspace{8pt}
\noindent {\bf (i)}	In this case we may assume, up to symmetry, that $l$ lies below $\ell_{12}$ as in Figure~\ref{path-fig1}.
Let $c$ be the intersection point of $\ell_{12}$ with $\ell_{45}$, and $d$ be the intersection point of $\ell_{45}$ with the right boundary of $H$. Since $K$ is a cup, $c$ lies below $K$ and hence in $H$. Consider the ray emanating from $p_4$ and passing through $c$. Rotate this ray clockwise around $p_4$ and stop as soon as hitting a point in the triangle $\bigtriangleup p_2cp_4$; see Figure~\ref{path-fig1}(a). Notice that such a point exists because $p_3$ is in $\bigtriangleup p_2cp_4$. Denote this first hit by $p'_3$ (it might be the case that $p'_3=p_3$). Then $p_1,p_2,p'_3,p_4,p_5$ is a $5$-cup which we denote by $K'$ (again, it might be the case that $K'=K$). Let $c'$ be the intersection point of the rotated ray with $\ell_{12}$. Our choice of $p'_3$ implies that the triangle $\bigtriangleup cp_4c'$ is empty, i.e. its interior has no points of $P$; this triangle is denoted by $E_3$ in Figure~\ref{path-fig1}(a).

\begin{figure}[htb]
	\centering
	\setlength{\tabcolsep}{0in}
	$\begin{tabular}{cc}
		\multicolumn{1}{m{.5\columnwidth}}{\centering\vspace{0pt}\includegraphics[width=.38\columnwidth]{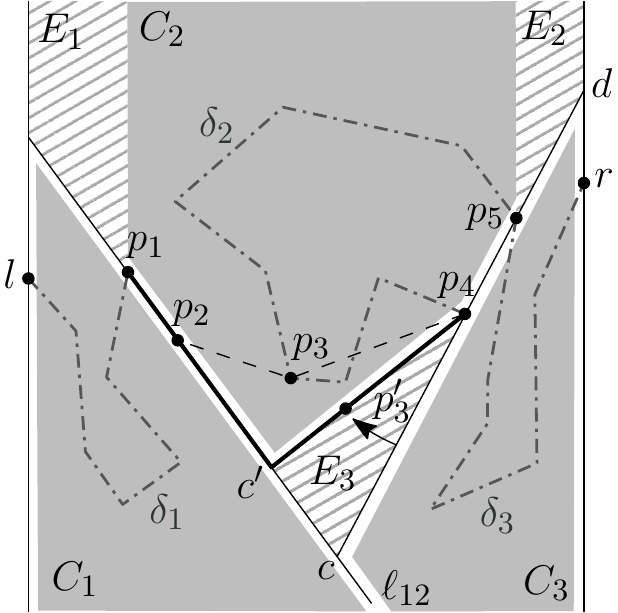}}
		&\multicolumn{1}{m{.5\columnwidth}}{\centering\vspace{0pt}\includegraphics[width=.38\columnwidth]{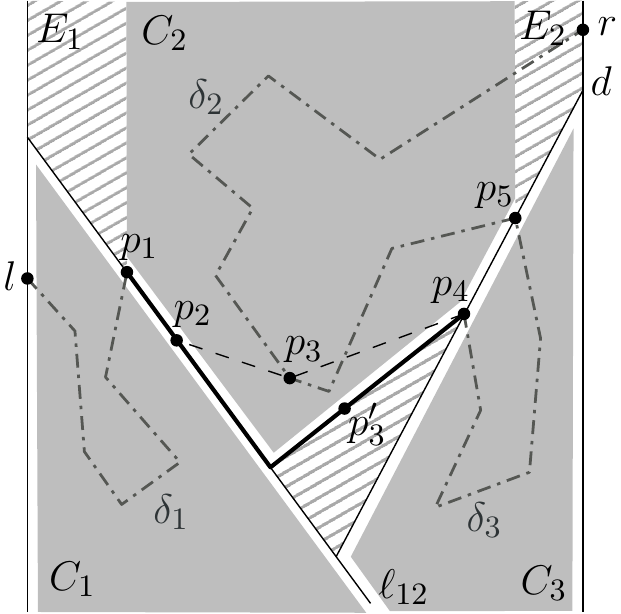}}
		\\
		(a)   &(b) 
	\end{tabular}$
	\caption{Illustration of the proof of Lemma~\ref{27-lemma}. (a) $l$ lies below $\ell_{12}$ and $r$ lies below $\ell_{45}$. (b) $l$ lies below $\ell_{12}$ and $r$ lies above $\ell_{45}$.}
	\label{path-fig1}
\end{figure}

The points of $P'$ lie in the interior or on the boundary of three convex regions $C_1,C_2,C_3$ as depicted in Figures~\ref{path-fig1}(a) and \ref{path-fig1}(b). The region $C_1$ is the intersection of $H$ and the halfplane below $\ell_{12}$.  The region $C_3$ is the intersection of $H$ and the halfplane above $\ell_{12}$ and the halfplane below $\ell_{45}$. The region $C_2$ is the intersection of $H$ and five halfplanes (the halfplanes above the lines $\ell_{12}$, $\ell_{45}$, $\ell(c',p_4)$, the halfplane to the right of the vertical line through $p_1$, and the halfplane to the left of the vertical line through $p_5$). Let $P_i$ be the set of points of $P$ in the interior (but not on the boundary) of each $C_i$.  Then $P_1\cup P_2\cup P_3= P\setminus\{l,p_1,p_2,p'_3,p_4,p_5,r\}$, and thus $|P_1|+|P_2|+|P_3|=|P|-7$.

We construct a covering path for $P$ as follows. The four points $p_1,p_2,p'_3,p_4$ can be covered by the path $(p_1,c',p_4)$ which has two edges $p_1c'$ and $c'p_4$. 
Let $\delta_1$ be the noncrossing path with $|P_1|+1$ edges that is obtained by applying Lemma~\ref{convex-lemma} on $(P_1, C_1, l, p_1)$ where $l$ and $p_1$ play the roles of $p$ and $q$ in the lemma. We now consider two subcases. 
	\begin{itemize}
		\item {\em $r$ lies below $\ell_{45}$.} In this case $r$ is on the boundary of $C_3$, as in Figure \ref{path-fig1}(a). 
		Let $\delta_2$ and $\delta_3$ be the noncrossing paths with $|P_2|+1$ and $|P_3|+1$ edges that are obtained by applying Lemma~\ref{convex-lemma} on $(P_2,C_2,p_4,p_5)$ and $(P_3, C_3,p_5,r)$, respectively. By interconnecting these paths we obtain a noncrossing covering path $\delta_1\oplus (p_1,c',p_4)\oplus \delta_2\oplus\delta_3$ for $P$. This path has $(|P_1|+1)+2+(|P_2|+1)+(|P_3|+1)=|P|-2$ edges, and it lies in $H$.
		
		\item {\em $r$ lies above $\ell_{45}$.} In this case $r$ is on the boundary of the convex region $C_2\cup E_2$, as in Figure \ref{path-fig1}(b). 
		Let $\delta_2$ and $\delta_3$ be the noncrossing paths obtained by applying Lemma~\ref{convex-lemma} on $(P_2,C_2\cup E_2,p_5,r)$ and $(P_3,C_3,p_4,p_5)$, respectively. Then $\delta_1\oplus (p_1,c',p_4)\oplus \delta_3\oplus\delta_2$ is a noncrossing covering path for $P$. This path has $|P|-2$ edges, and it lies in $H$. 
	\end{itemize}	
	
\begin{figure}[htb]
	\centering
	\setlength{\tabcolsep}{0in}
	$\begin{tabular}{cc}
		\multicolumn{1}{m{.6\columnwidth}}{\centering\vspace{0pt}\includegraphics[width=.43\columnwidth]{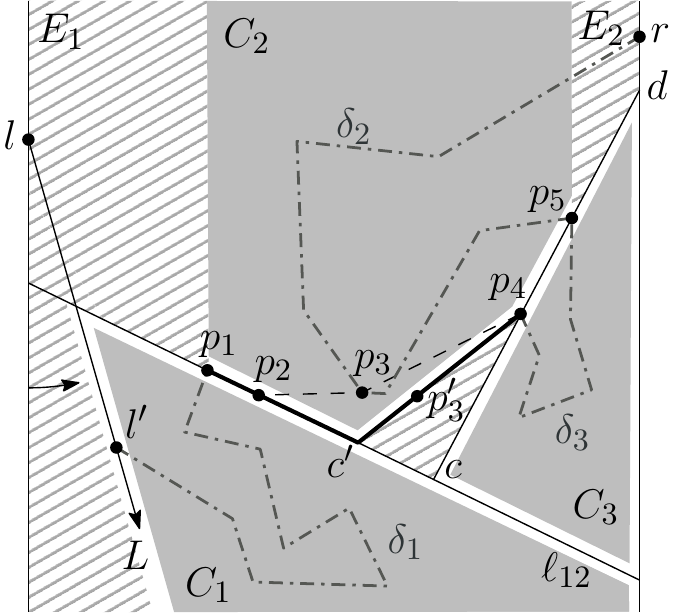}}
		&\multicolumn{1}{m{.4\columnwidth}}{\centering\vspace{0pt}\includegraphics[width=.27\columnwidth]{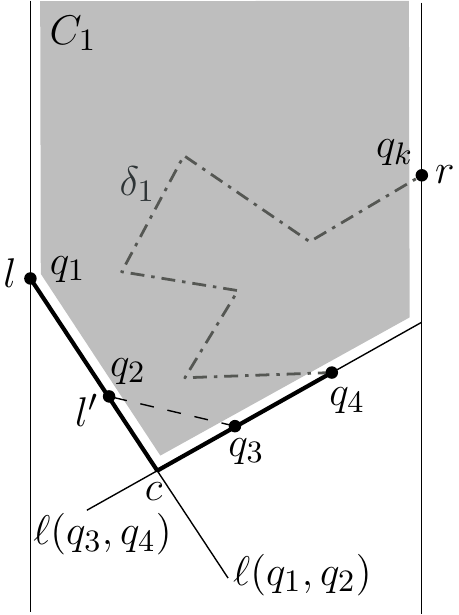}}
		\\
		(a)   &(b) 
	\end{tabular}$		
	\caption{Illustration of the proof of Lemma~\ref{27-lemma} where $l$ lies above $\ell_{12}$ and $r$ lies above $\ell_{45}$. (a) $l'\neq p_1$, and (b) $l'=p_1$ and $r'=p_5$ (here $p_1=q_2$ and $p_5=q_{k-1}$).}
	\label{path-fig2}
\end{figure}	
	
\vspace{8pt}
\noindent {\bf (ii)} In this case $l$ lies above $\ell_{12}$ and $r$ lies above $\ell_{45}$.
Let $L$ and $R$ be the downward rays emanating from $l$ and $r$, respectively. Rotate $L$ counterclockwise around $l$ and stop as soon as hitting a point $l'$ of $P$. Since $E_1$ is empty, $l'$ is either $p_1$ or a point below $\ell_{12}$; see Figure~\ref{path-fig2}(a). Rotate $R$ clockwise around $r$ and stop as soon as hitting a point $r'$ of $P$. Since $E_2$ is empty, $r'$ is either $p_5$ or a point below $\ell_{45}$. We distinguish two subcases.

\begin{itemize}
	\item {\em $l'\neq p_1$ or $r'\neq p_5$.} Up to symmetry we assume that $l'\neq p_1$ as depicted in Figure~\ref{path-fig2}(a).
	Define $c$, $c'$, $d$, $p'_3$ and the $5$-cup $K'$ as in case (i), and recall that the triangle $\bigtriangleup cp_4c'$ is empty.
	The points of $P'$ lie in the interior or on the boundary of three convex regions $C_1,C_2,C_3$ as depicted in Figures~\ref{path-fig2}(a). The region $C_1$ is the intersection of $H$ and the halfplane below $\ell_{12}$ and the halfplane above $\ell(l,l')$.  The regions $C_2$ and $C_3$ are defined as in case (i). Let $P_i$ be the set of points of $P$ in the interior (but not on the boundary) of each $C_i$.  Then $P_1\cup P_2\cup P_3= P\setminus\{l,l',p_1,p_2,p'_3,p_4,p_5,r\}$, and thus $|P_1|+|P_2|+|P_3|=|P|-8$.
	
	We cover $l$ and $l'$ by the edge $(l,l')$ and cover the four points $p_1,p_2,p'_3,p_4$ by the path $(p_1,c',p_4)$ which has two edges.
	Let $\delta_1$, $\delta_2$, and $\delta_3$ be the noncrossing paths with $|P_1|+1$, $|P_2|+1$, and $|P_3|+1$ edges obtained by applying Lemma~\ref{convex-lemma} on $(P_1,C_1,l',p_1)$, $(P_2,C_2\cup E_2,p_5,r)$, and $(P_3,C_3,p_4,p_5)$, respectively; see Figures~\ref{path-fig2}(a). By interconnecting these paths we obtain a noncrossing covering path $(l,l')\oplus\delta_1\oplus (p_1,c',p_4)\oplus \delta_3\oplus\delta_2$ for $P$. This path has $1+(|P_1|+1)+2+(|P_3|+1)+(|P_2|+1)=|P|-2$ edges, and it lies in $H$.
	
	\item {\em $l'= p_1$ and $r' = p_5$.} In this case the lower chain on the boundary of $\conv{P}$ has at least $5$ vertices, including $l$, $l'$, $r$, $r'$, and a point in the triangle formed by $L$, $R$, and $\ell(p_1,p_5)$. Let $k\geqslant 5$ be the number of vertices of this chain. Let $q_1, q_2,\dots, q_k$ denote the vertices of this chain that appear in this order from left to right, as in Figure~\ref{path-fig2}(b). Then $q_1=l$, $q_2=l'=p_1$, $q_k=r$, and $q_{k-1}=r'=p_5$. 
	
	Let $c$ be the intersection point of $\ell(q_1,q_2)$ and $\ell(q_3,q_4)$, which lies in $H$.
Then, the four points $q_1,q_2,q_3,q_4$ can be covered by the path $(q_1,c,q_4)$. 
All points of $P$ lie in the interior or on the boundary of a convex region $C_1$ that is the intersection of $H$ with the halfplanes above $\ell(q_1,q_2)$ and $\ell(q_3,q_4)$; this region is shaded in Figure~\ref{path-fig2}(b). Let $P_1$ be the points of $P$ that lie in the interior (but not on the boundary) of $C_1$. Then $P_1=P\setminus\{q_1,q_2,q_3,q_4,q_k\}$ and $|P_1|=|P|-5$. Let $\delta_1$ be the covering path with $|P_1|+1$ edges that is obtained by applying Lemma~\ref{convex-lemma} on $(P_1, C_1, q_4,q_k)$ where $q_4$ and $q_k$ play the roles of $p$ and $q$ in the lemma. Then $(q_1,c,q_4)\oplus\delta_1$ is a noncrossing covering path for $P$. This path has $2+(|P_1|+1)=|P|-2$ edges, and it lies in $H$.
\end{itemize}
{\color{mycolor}
	This is the end of our proof (for $P'$ being in general position). %To deal with collinearities of three or more points we adjust the proof as follows. Among all the $5$-cups we take $K$ to be one that has a minimum $x$-span. This ensures that no other point of $P'$ lies on the (weakly) convex chain $p_1,p_2,p_3,p_4,p_5$. If the rotating ray around $p_4$ hits more than one points at the same time then we take $p'_3$ to be the one closest to $p_4$. If $p'_3$ is collinear with $p_1$ $p_2$ then we take $c'$ to be the same as $p'_3$. We consider any point that lies on edges $p_1c$, $cp_4$ (in the first case) and $p_1c'$, $c'p_4$ (in the second case)as being covered by these edges; this would lead to a smaller covering path. If $l$ lies on $\ell_{12}$ we consider handle it as in case (i) where $l$ is below $\ell_{12}$. We deal with $r$ and $\ell_{45}$ similarly. To define the sets $P_i$, all points that lie between boundaries of two convex regions will be assigned to only one of the regions.

One can simply adjust the above construction to work even if $P'$ is not in general position. For the sake of completeness here we give a brief description of an alternative (and perhaps simpler) construction when $P'$ has three or more collinear points. Let $p_1$, $p_2$, $p_3$ be three collinear points in $P'$ from left to right and let $\ell_{13}$ be the line through these points. We choose $p_1,p_2,p_3$ in such a way that there is no point of $P'$ on $\ell_{13}$ to the left of $p_1$ or to the right of $p_3$. Up to symmetry we have two cases: (i) $l$ lies on or above $\ell_{13}$ and $r$ lies on or below $\ell_{13}$, and (ii) both $l$ and $r$ lie below $\ell_{13}$. 

In case (i) we first obtain a path by applying Lemma~\ref{convex-lemma} on $l$, $p_1$ and all points above $\ell_{13}$. Then we connect $p_1$ and $p_3$ by one edge which also covers $p_2$. Then we obtain another path by applying Lemma~\ref{convex-lemma} on $r$, $p_3$ and all points below $\ell_{13}$. This gives a a covering path with $|P|-2$ edges.

In case (ii) we start from $l$ and walk on the vertices of $\conv{P}$ in clockwise direction (and at the same time cover the visited vertices) and stop at the first vertex, say $p_0$, for which the next vertex, say $x$, is on or above  $\ell_{13}$ (it could be the case that $p_0=l$). Denote the traversed path between $l$ and  $p_0$ by $\delta_l$. First assume that  $x$ is above $\ell_{13}$. We connect $p_0$ to $x$. Then we obtain a path by applying Lemma~\ref{convex-lemma} on $x$, $p_1$ and all points above $\ell_{13}$. Then we connect $p_1$ to $p_3$ by one edge which also covers $p_2$. Then we extend the current path to a covering path for $P$ by applying Lemma~\ref{convex-lemma} on $p_3$, $r$ and the remaining points below $\ell_{13}$. Now assume that $x$ is on $\ell_{13}$, in which case $x=p_1$. If there is no point of $P'$ above $\ell_{13}$ then we connect $p_1$ to $p_3$ by one edge and then extend it to a covering path for $P$ by applying Lemma~\ref{convex-lemma} on $p_3$, $r$ and the remaining points below $\ell_{13}$. Assume that there are points above $\ell_{13}$. We repeat the above process from $r$ by a counterclockwise walk on the vertices of $\conv{P}$, and due to symmetry, assume that $p_3$ is the first visited vertex that lies on or above $\ell_{13}$. Let $p_4$ denote the vertex of $\conv{P}$ after $p_3$. Notice that $p_4$ lies above $\ell_{13}$. Let $c$ be the intersection point of the lines $\ell(p_0,p_1)$ and $\ell(p_3,p_4)$. To obtain a covering path, we start with $\delta_l$, connect its endpoint $p_0$ to $c$, and connect $c$ to $p_3$; these two edges cover $p_1$ and $p_4$. Then we continue by a path obtained from Lemma~\ref{convex-lemma} applied on $p_3$, $r$ and the remaining points.} 
\end{proof}

The following corollary, although very simple, will be helpful in the analysis of our algorithm.

\begin{corollary}
	\label{26-cor}
	Let $Q$ be a set of at least $22$ points in the plane and let $l$ be its leftmost point. Then there exists a noncrossing covering path for $Q$ with $|Q|-2$ edges that lies to the right of the vertical line through $l$ and has $l$ as an endpoint.
\end{corollary}
\begin{proof}
	We add a dummy point $r$ to the right of all points in $Q$. Let $P=Q\cup\{r\}$. We obtain a noncrossing covering path $\delta$ for $P$ with $|P|-2$ edges by Lemma~\ref{27-lemma}.	
	{\color{mycolor}Recall that $r$ is an endpoint of $\delta$. Also recall from the proof of Lemma~\ref{27-lemma} that in all cases $r$ gets connected to $\delta$ by a path that is obtained from Lemma~\ref{convex-lemma}. No edge of this path has a point of $P$ in its interior (even if two consecutive edges happen to be collinear, we treat them as two different edges). Thus, the edge of $\delta$ that covers $r$ has no point in its interior.}
	%Recall that $r$ is an endpoint of $\delta$. Also recall from the proof of Lemma~\ref{27-lemma} that only the edges $p_1c$, $cp_4$ (in the first case) and $p_1c'$, $c'p_4$ (in the second case) have points of $P$ in their interior. Thus, the edge of $\delta$ that covers $r$ has no point in its interior. 
	Therefore, by removing $r$ and its incident edge from $\delta$ we obtain a covering path with $|Q|-2$ edges for $Q$ that satisfies the conditions of the corollary.
\end{proof}

\begin{theorem}
	\label{path-thr}
	Every set of $n$ points in the plane admits a noncrossing covering path
	with at most $\lceil21n/22\rceil -1$ edges. Thus, $\pi(n)\leqslant (1-1/22)n$. Such a path can be computed in $O(n\log n)$ time.
\end{theorem}
\begin{proof}
	Let $P$ be a set of $n$ points in the plane. After a suitable rotation we may assume that no two points of $P$ have the same $x$-coordinate.
	Draw vertical lines in the plane such that each line goes through a point of $P$, there are exactly 21 points of $P$ between any pair of consecutive lines, no point of $P$ lies to the left of the leftmost line, and at most 21 points of $P$ lie to the right of the rightmost line; see Figure~\ref{path-fig3}. Each pair of consecutive lines defines a vertical strip containing 23 points; 21 points in its interior and 2 points on its boundary (the point on the boundary of two consecutive strips is counted for both strips). For the 23 points in each strip we obtain a noncrossing covering path with 21 edges using Lemma~\ref{27-lemma}. Each path lies in its corresponding strip and its endpoints are the two points on the boundary of the strip. By assigning to each strip the point on its left boundary, it turns out that for every 22 points we get a path with 21 edges.
	
	\begin{figure}[htb]
		\centering
		\setlength{\tabcolsep}{0in}
		\includegraphics[width=.65\columnwidth]{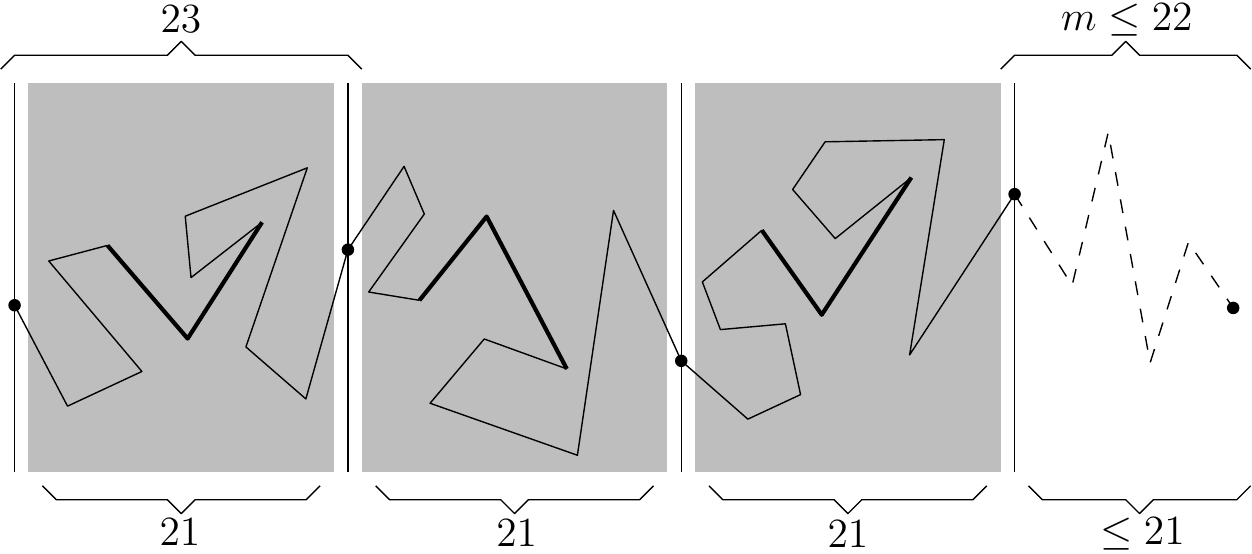}
		\caption{Illustration of the proof of Theorem~\ref{path-thr}.}
		\label{path-fig3}
	\end{figure}

	Let $m$ be the number of points on or to the right of the rightmost line, and notice that $m\leqslant 22$. We distinguish between two cases $m=22$ and $m<22$.
	
	If $m=22$ (in this case $n$ is divisible by $22$) then we cover these $22$ points by a noncrossing path with $20$ edges using Corollary~\ref{26-cor}. The union of this path and the paths constructed within the strips is a noncrossing covering path for $P$. The total number of edges in this path is $21n/22-1$.
	
	 If $m<22$ then $m=n-22\lfloor n/22\rfloor$. In this case we cover the $m$ points by an $x$-monotone path with $m-1$ edges (dashed segments in Figure~\ref{path-fig3}). Again, the union of this path and the paths constructed within the strips is a noncrossing covering path for $P$. The total number of edges in this path is $21\lfloor n/22\rfloor + m -1 =n-\lfloor n/22\rfloor -1=\lceil 21n/22\rceil -1$.
	 
	 Each call to Lemma~\ref{27-lemma} and Corollary~\ref{26-cor} takes constant time. Therefore, after rotating and sorting the points in $O(n\log n)$ time, the rest of the algorithm takes linear time.
\end{proof}

Our path construction in Theorem~\ref{path-thr} achieves a similar bound for covering cycles.

\begin{corollary}
	\label{cycle-cor}
	Every set of $n$ points in the plane admits a noncrossing covering cycle
	with at most $\lceil21n/22\rceil +1$ edges. Such a cycle can be computed in $O(n\log n)$ time.
\end{corollary}
\begin{proof}
Let $\delta$ be the path constructed by Theorem~\ref{path-thr} on a point set $P$ of size $n$. Recall $m$ from the proof of this theorem. If $m<22$  then the two endpoints of $\delta$ are the leftmost and rightmost points of $P$. Thus, by introducing a new point $p$ with a sufficiently large $y$-coordinate and connecting it to the two endpoints of $\delta$, we obtain a noncrossing covering cycle for $P$. If $m=22$ then the dummy point that was introduced in Corollary~\ref{26-cor} could be chosen suitably to play the role of $p$. 
\end{proof}

\section{Noncrossing Covering Trees}

In this section we prove the following theorem which gives an algorithm for computing a noncrossing covering tree with roughly $4n/5$ edges. We should clarify that the number of {\em edges} of a tree is different from the number of its {\em segments} (where each segment is either a single edge or a
chain of several collinear edges of the tree). For example the tree in Figure~\ref{rainbow-fig}(b) has 10 edges and 7 segments, where the segments $p_1p_7$ and $p_5p_8$ consist of 3 and 2 collinear edges, respectively.

\begin{theorem}
	\label{tree-thr}
	Every set of $n$ points in the plane admits a noncrossing covering tree
	with at most $\lceil 4n/5\rceil$ edges. Thus, $\tau(n)\leqslant \lceil 4n/5\rceil$. Such a tree can be computed in $O(n\log n)$ time.
\end{theorem}
\begin{proof}
	Let $P$ be a set of $n$ points in the plane. After a suitable rotation we may assume that no two points of $P$ have the same $x$-coordinate.
	We present an iterative algorithm to compute a noncrossing covering tree for $P$ that consists of at most $\lceil4n/5\rceil$ edges. In a nutshell, the algorithm scans the points from left to right and in every iteration (except possibly the last iteration) it considers $4$ or $5$ new points and covers them with $3$ or $4$ new edges, respectively. Thus the ratio of the number of new edges to the number of covered points would be at most $4/5$. 
	We begin by describing an intermediate iteration of the algorithm; the first and last iterations will be described later. {\color{mycolor} We assume that the scanned points in each iteration are in general position. In the end of the proof we describe how to handle collinearities.}	
	Let $m$ be the number of points that have been scanned so far and let $l$ be the rightmost scanned point (our choice of the letter $l$ will become clear shortly).
	We maintain the following invariant at the beginning of every intermediate iteration. 
	
\vspace{8pt}
\begin{center}
\begin{minipage}[c]{0.93\textwidth}
{\bf Invariant. } All the $m$ points that have been scanned so far, are covered by a noncrossing tree $T$ with at most $4\lfloor m/5\rfloor$ edges. The tree $T$ lies to the left of the vertical line through $l$ and the degree of $l$ in $T$ is one.
\end{minipage}
\end{center}
\vspace{8pt}

In the current (intermediate) iteration we scan four new points, namely $a$, $b$, $c$, and $r$ where $r$ is the rightmost point. Let $H$ be the vertical strip bounded by the vertical lines through $l$ and $r$ ($l$ is the leftmost point and $r$ is the rightmost point in $H$); see Figure~\ref{tree-fig}(a). Let $Q=\{l,a,b,c,r\}$. We consider three cases depending on the number of vertices of $\conv{Q}$. Notice that $r$ and $l$ are two vertices of $\conv{Q}$.

\begin{itemize}
	\item {\em $\conv{Q}$ has three vertices.} Let $a$ be the third vertex of $\conv{Q}$. Then $b$ and $c$ lie in the interior of $\conv{Q}$, as in Figure~\ref{tree-fig}(a). In this case two vertices of $\conv{Q}$, say $l$ and $r$, lie on the same side of $\ell(b,c)$. Thus $l$, $b$, $c$, and $r$ form a convex quadrilateral. After a suitable relabeling assume that $l,b,c,r$ appear in this order along the boundary of the quadrilateral. Let $x$ be the intersection point of $\ell(l,b)$ and $\ell(r,c)$, which lies in the triangle $\bigtriangleup lra$. We cover the four scanned points $a$, $b$, $c$, and $r$ by three edges $xl$, $xr$, and $xa$ which lie in $H$. We add these edges to $T$. The degree of $r$ is one in the new tree (no matter which two vertices of $\conv{Q}$ lay on the same side of $\ell(b,c)$). The invariant holds and we proceed to the next iteration.

	\begin{figure}[htb]
		\centering
		\setlength{\tabcolsep}{0in}
		$\begin{tabular}{ccccc}
			\multicolumn{1}{m{.18\columnwidth}}{\centering\vspace{0pt}\includegraphics[width=.17\columnwidth]{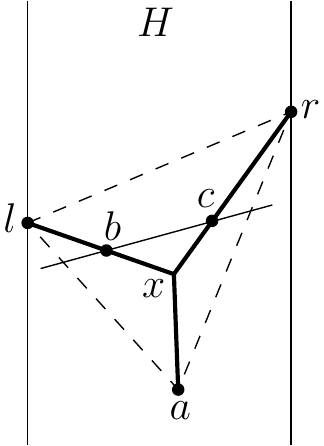}}
			&\multicolumn{1}{m{.18\columnwidth}}{\centering\vspace{0pt}\includegraphics[width=.17\columnwidth]{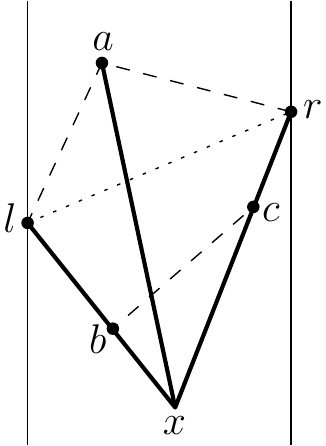}}
			&\multicolumn{1}{m{.2\columnwidth}}{\centering\vspace{0pt}\includegraphics[width=.19\columnwidth]{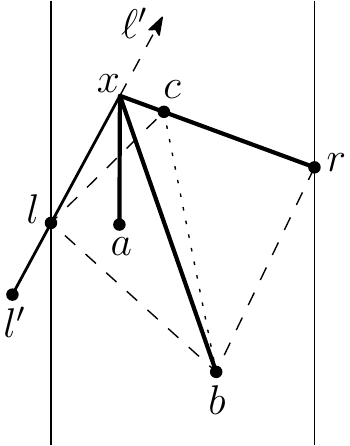}}
			&\multicolumn{1}{m{.2\columnwidth}}{\centering\vspace{0pt}\includegraphics[width=.19\columnwidth]{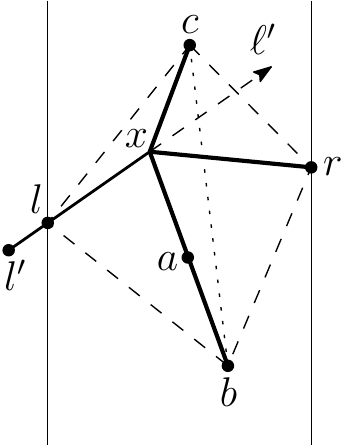}}
			&\multicolumn{1}{m{.24\columnwidth}}{\centering\vspace{0pt}\includegraphics[width=.23\columnwidth]{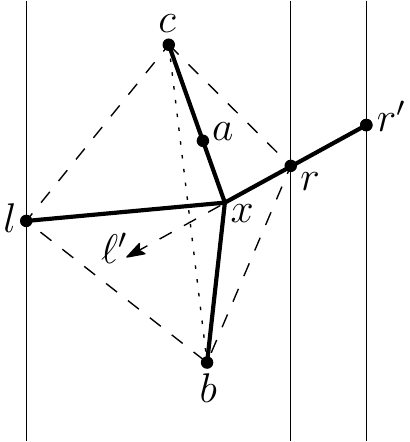}}
			\\
			(a)   &(b) &(c)&(d)&(e)
		\end{tabular}$
		\caption{Illustration of the proof of Theorem~\ref{tree-thr}. (a) $\conv{Q}$ has three vertices. (b) $\conv{Q}$ has five vertices. (c)-(e) $\conv{Q}$ has four vertices.}
		\label{tree-fig}
	\end{figure}

	{\em \item $\conv{Q}$ has five vertices.} We explain this case first as our argument is shorter, and also it will be used for the next case. In this case $Q$ contains a $4$-cap or a $4$-cup with endpoints $l$ and $r$. After a suitable reflection and relabeling assume it has the $4$-cup $l,b,c,r$ as in Figure~\ref{tree-fig}(b). Let $x$ be the intersection point of $\ell(l,b)$ and $\ell(r,c)$, and observe that it lies in $H$. We cover $a$, $b$, $c$, and $r$ by three edges $xl$, $xr$, and $xa$ which lie in $H$. We add these edges to $T$. The degree of $r$ is one in the new tree. The invariant holds for the next iteration. 
	{\em \item $\conv{Q}$ has four vertices.} After a suitable relabeling assume that $b$ and $c$ are two vertices of $\conv{Q}$ (other than $l$ and $r$). Thus $a$ lies in the interior of $\conv{Q}$. If both $b$ and $c$ lie above or below $\ell(l,r)$ then $l,b,c,r$ form a $4$-cap or a $4$-cup, in which case we cover the points as in the previous case. Therefore we may assume that one point, say $b$, lies below $\ell(l,r)$ and $c$ lies above $\ell(l,r)$ as in Figures~\ref{tree-fig}(c)-(e). We consider two subcases.
	
	\begin{itemize}
		\item {\em $a$ lies in the triangle $\bigtriangleup lbc$.} By the invariant, $l$ has degree one in $T$. Let $l'$ be the neighboring vertex of $l$ in $T$. Let $\ell'$ be the ray emanating from $l'$ and passing through $l$. We consider two subcases: (i) the segment $bc$ does not intersect $\ell'$ and (ii) the segment $bc$ intersects $\ell'$.
		
		In case (i) the segment $bc$ lies below or above $\ell'$. By symmetry assume that it lies below $\ell'$. Then $a$ and $r$ also lie below $\ell'$, as in Figure~\ref{tree-fig}(c). In this case $\ell(r,c)$ intersects $\ell'$. Let $x$ be their intersection point, and observe that it lies in $H$. We replace the edge $l'l$ of $T$ by $l'x$ (this does not increase the number of edges because $l$ has degree one). Notice that $l'x$ contains $l$. Then we cover $a$, $b$, $c$, and $r$ by adding three edges $xr$, $xb$, and $xa$ to $T$.   Therefore the number of edges of $T$ is increased by $3$. Moreover, $r$ has degree one in the new tree, and all the newly introduced edges lie to the left of the vertical line through $r$. Thus the invariant holds for the next iteration.

		In case (ii) the ray $\ell'$ goes through $\bigtriangleup lbc$. The point $a$ lies below or above $\ell'$. By symmetry assume that it lies below $\ell'$, as in Figure~\ref{tree-fig}(d). Let $x$ be the intersection point of $\ell(a,b)$ and $\ell'$, which lies in $\bigtriangleup lbc$. We replace the edge $l'l$ of $T$ by $l'x$. Then we cover $a$, $b$, $c$, and $r$ by adding three edges $xr$, $xb$, and $xc$ to $T$. Thus, the number of edges of $T$ is increased by $3$, the vertex $r$ has degree one in the new tree, and all new edges lie to the left of the vertical line through $r$. The invariant holds for the next iteration.

		\item {\em $a$ lies in the triangle $\bigtriangleup rbc$.} Here is the place where we use four new edges to cover five vertices. In fact the ratio $4/5$ comes from this case (In previous cases we were able to cover four points by three new edges). In this case we scan the next point after $r$ which we denote by $r'$, as in Figure~\ref{tree-fig}(e). Now let $\ell'$ be the ray emanating from $r'$ and passing through $r$. The current setting is essentially the vertical reflection of the previous case where $r$ and $r'$ play the roles of $l$ and $l'$, respectively. We handle this case analogous to the previous case. Our analysis is also analogous except that now we consider the edge $r'x$ as a new edge. Thus we use four new edges to cover five points $a$, $b$, $c$, $r$, and $r'$. All new edges lie to the left of the vertical line through $r'$, and the degree of $r'$ is one in the new tree. Thus the invariant holds for the next iteration. 
	\end{itemize}
\end{itemize}
This is the end of an intermediate iteration. The noncrossing property of the resulting tree follows from our construction.  This iteration suggests a covering tree with roughly $4n/5$ edges. To get the exact claimed bound we need to have a closer look at the first and last iterations of the algorithm.

For the first iteration of the algorithm we scan only the leftmost input point. This point will play the role of $l$ for the second iteration (which is the first intermediate iteration). The invariant holds for the second iteration because the tree has no edges at this point. If we happen to use the edge $l'l$ in the second iteration, then we take $l'=l$ and give the ray $\ell'$ an arbitrary direction to the right. Based on the above construction this could happen only when we scan four points ($a$, $b$, $c$, $r$) in the second iteration. In this case the first five points ($l$, $a$, $b$, $c$, $r$) are covered by four edges, and thus the invariant holds for the following iteration.
In the last iteration of the algorithm we are left with $w\leqslant 4$ points that are not being scanned. We connect these $w$ points by $w$ edges to the rightmost scanned point. 
Therefore, the algorithm covers all points by a noncrossing tree with at most $\lceil4n/5\rceil$ edges.

{\color{mycolor} If three or more of the scanned points are collinear then cover all collinear points by one edge and connect the left endpoint of this edge to $l$. Then we connect every remaining scanned point to $l$. The number of new edges is at most 3 (for 4 scanned points) and 4 (for 5 scanned points).}

Each iteration takes constant time. Therefore, after rotating and sorting the points in $O(n\log n)$ time, the rest of the algorithm takes linear time.
\end{proof}

\section{Perfect Rainbow Polygons}
\label{polygon-section}
Recall that a perfect rainbow polygon for a set of colored points, is a simple polygon that contains exactly one point of each color in its interior or on its boundary. Figure~\ref{rainbow-fig}(a) shows a perfect rainbow polygon of size 9 (nine edges) for an 8-colored point set (i.e. colored by 8 different colors).  There is a relation (as described below) between rainbow polygons and noncrossing covering trees. We employ this relation (similar to \cite{Flores-Penaloza21}) and present an algorithm that achieves a perfect rainbow polygon of size at most $7k/5+O(1)$ for any $k$-colored point set. 
%This improves the previous bound $10k/7+O(1)$ from \cite{Flores-Penaloza21}.

\begin{figure}[htb]
	\centering
	\setlength{\tabcolsep}{0in}
	$\begin{tabular}{cc}
		\multicolumn{1}{m{.35\columnwidth}}{\centering\vspace{0pt}\includegraphics[width=.27\columnwidth]{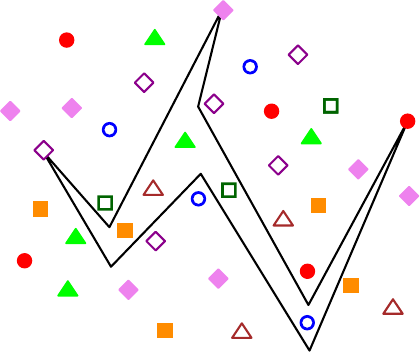}}
		&\multicolumn{1}{m{.65\columnwidth}}{\centering\vspace{0pt}\includegraphics[width=.6\columnwidth]{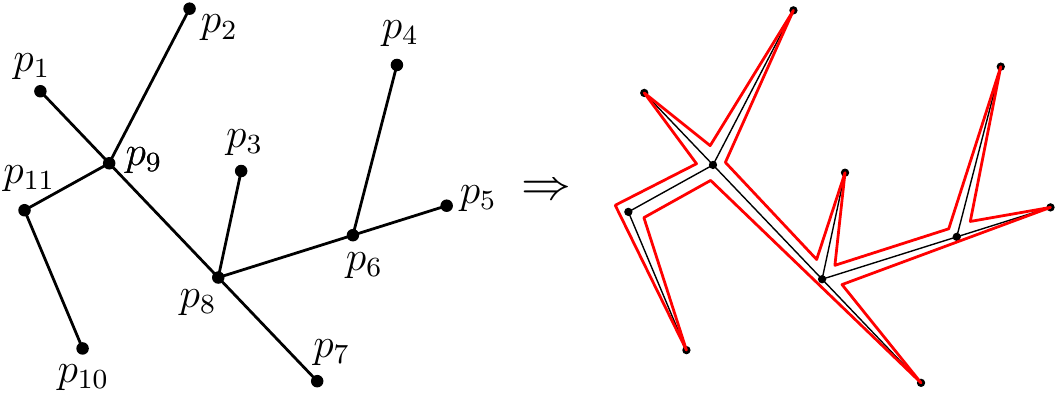}}
		\\
		(a) &(b) 
	\end{tabular}$
	\caption{(a) A perfect rainbow polygon of size $9$ for an $8$-colored point set. (b) Left: A noncrossing tree with ten edges that can be partitioned into seven segments $\{p_1p_7,p_2p_9,p_3p_8, p_4p_6,p_5p_8,\allowbreak p_{10}p_{11},\allowbreak p_{11}p_9\}$;
		$p_6$ and $p_8$ are two forks with multiplicity $1$ and $p_9$ is a fork with multiplicity $2$. Right: Obtaining a simple enclosing polygon from the tree.}
	\label{rainbow-fig}
\end{figure}

We adopt the following notation and definitions from \cite{Flores-Penaloza21}.
Let $T$ be a noncrossing geometric
tree. Recall that a segment of $T$ is a chain of collinear edges in $T$. 
Let $M$ be a partition of the edges of $T$ into a minimal number of pairwise noncrossing segments. Let $s$ denote the number of segments in $M$. A {\em fork} of $T$ (with respect to $M$) is a vertex $f$ that lies in the interior of a segment $ab\in M$ and it is an endpoint of another segment of $M$. The {\em multiplicity} of $f$ is a number in $\{1,2\}$ that is determined as follows. If the segments that have $f$ as an endpoint lie on both sides of $\ell(a,b)$ then $f$ has multiplicity 2, otherwise (the segments lie on one side of the line) $f$ has multiplicity 1. See the tree in Figure~\ref{rainbow-fig}(b) for an example. Let $t$ denote the sum of multiplicities of all forks in $T$.
The following lemma expresses the size of
a polygon enclosing $T$ in terms of $s$ and $t$. 
%As noted in \cite{Flores-Penaloza21}, any vertex of degree two that is incident to two collinear edges can be suppressed, and one may assume that $T$ has no such vertices.

\begin{lemma}[Flores{-}Pe{\~{n}}aloza~\etal~\cite{Flores-Penaloza21}]
\label{tree-polygon-lemma}
 Let $T$ be a noncrossing geometric tree and $M$ be a partition of its edges into a minimal number of pairwise noncrossing segments. Let $s$ be the number of segments in $M$ and $t$ be the total multiplicity of forks in $T$. If $s\geqslant 2$ and $t\geqslant 0$, then for every $\varepsilon > 0$ there exists a simple polygon of size $2s + t$ and of area at most $\varepsilon$ that encloses $T$.
\end{lemma}

There are simple intuitions behind Lemma~\ref{tree-polygon-lemma}. For example if we cut out the tree $T$ from the plane, then the resulting hole could be expressed as a desired polygon. Alternatively, if we start from a vertex of $T$ and walk around $T$ (arbitrary close to its edges) until we come back to the starting vertex, then the traversed tour could be represented as a desired polygon. See Figure~\ref{rainbow-fig}(b).

In view of Lemma~\ref{tree-polygon-lemma}, a better covering tree (i.e. for which $2s+t$ is smaller) leads to a better polygon (i.e. with fewer edges). 
In Theorem~\ref{tree-polygon-thr} (proven in Section~\ref{better-tree-section}) we show that any set of $k$ points in the plane in general position admits a noncrossing covering tree for which $s\leqslant\lceil \frac{3k}{5}\rceil +2$ and  $t\leqslant\lceil \frac{k}{5}\rceil$.
With this lemma and theorem in hand, we present our algorithm for computing a perfect rainbow polygon.

\paragraph{Algorithm {\normalfont (in a nutshell)}.}
Let $P$ be a set of $n$ points in the plane in general position that are colored by $k$ distinct colors. The algorithm picks one point from each color (arbitrarily), covers the chosen points by a noncrossing tree (using Theorem~\ref{tree-polygon-thr}), and then obtains a perfect rainbow polygon from the tree (using Lemma~\ref{tree-polygon-lemma}).

\paragraph{Analysis.}Let $K$ be the set of $k$ chosen points, and let $T$ be the covering tree for $K$ obtained by Theorem~\ref{tree-polygon-thr}. Then $s\leqslant\lceil \frac{3k}{5}\rceil +2 $ and $t\leqslant\lceil \frac{k}{5}\rceil$. Thus, the perfect rainbow polygon obtained by Lemma~\ref{tree-polygon-lemma} has size 
\[2s+t\leqslant 2\left(\left\lceil \frac{3k}{5}\right\rceil +2\right)+\left\lceil \frac{k}{5}\right\rceil\leqslant \left\lceil \frac{7k}{5}\right\rceil +6.\] 
The tree $T$ can be obtained in $O(k\log k)$ time, by Theorem~\ref{tree-polygon-thr}. To obtain a polygon (avoiding points of $P\setminus K$) from $T$ we need to choose a suitable $\varepsilon$ in Lemma~\ref{tree-polygon-lemma}. As noted in \cite{Flores-Penaloza21}, half of the minimum distance between the edges of $T$ and the points of $P\setminus K$ is a suitable $\varepsilon$, which can be found in $O(n\log n)$ time by computing the Voronoi diagram of the edges of $T$ together with the points of $P\setminus K$. Thus the total running time of the algorithm is $O(n\log n)$.
The following theorem summarizes our result in this section.

\begin{theorem}
	\label{polygon-thr}
	Every $k$-colored point set of size $n$ in the plane in general position admits a perfect rainbow polygon of size at most $\lceil 7k/5\rceil + 6$. Thus, $\rho(k)\leqslant \lceil 7k/5\rceil + 6$. Such a polygon can be computed in $O(n \log n)$ time.
\end{theorem}

\paragraph{Remark.} The general position assumption is necessary for our algorithm because if a non-selected point (i.e. a point of $P\setminus K$) lies on a segment of $T$ then the resulting polygon is not a valid rainbow polygon as it contains two or more points of the same color.
\subsection{A Better Covering Tree}
\label{better-tree-section}
Recall parameters $s$ and $t$ from the previous section. In this section we construct a covering tree for which $2s+t$ is smaller (compared to that of \cite{Flores-Penaloza21}). During the construction we will illustrate (in Figure~\ref{polygon-fig}) the structure of the polygon that is being obtained from the tree; this helps the reader to see that the polygon obtains 7 edges for every 5 points, and thus verify the bound $7k/5+O(1)$ intuitively. Our construction shares some similarities with the construction in our proof of Theorem~\ref{tree-thr}. However, the details of the two constructions are different because they have different objectives. We describe the shared parts briefly.

\begin{theorem}
	\label{tree-polygon-thr}
Let $K$ be a set of $k$ points in the plane in general position. Then, in $O(k \log k)$ time, one can construct a noncrossing covering tree for $K$ consisting of at most $\lceil \frac{3k}{5}\rceil +2 $ pairwise noncrossing segments with at most $\lceil \frac{k}{5}\rceil$ forks of multiplicity $1$.
\end{theorem}

\begin{proof}
	After a suitable rotation assume that no two points of $K$ have the same $x$-coordinate.
	We present an iterative algorithm that scans the points from left to right. In every iteration (except possibly the last iteration) we scan $5$ new points and cover them with $3$ new segments and $1$ new fork with multiplicity $1$. 
	As before, we start by describing an intermediate iteration of the algorithm; the first and last iterations will be described later. Let $m$ be the number of points scanned so far and let $l'$ be the rightmost scanned point.
	We maintain the following invariant at the beginning of every intermediate iteration. 
	
	\vspace{8pt}
	\begin{center}
		\begin{minipage}[c]{0.93\textwidth}
			{\bf Invariant. } All the $m$ points that have been scanned so far, are covered by a noncrossing tree $T$ with at most $3\lfloor (m-2)/5\rfloor+2$ segments with $\lfloor (m-2)/5\rfloor+1$ forks of multiplicity $1$. The tree $T$ lies to the left of the vertical line through $l'$ and the degree of $l'$ in $T$ is one.
		\end{minipage}
	\end{center}
	\vspace{8pt}
	
In the first iteration, which will be described later, we add to $T$ a long vertical segment through the leftmost point of $K$ such that the extension of any other segment of $T$ hits this segment.
	
	In the current (intermediate) iteration we scan the next five points, namely $l$, $a$, $b$, $c$, and $r$ where among them $l$ is the leftmost and $r$ is the rightmost. Let $H$ be the vertical strip bounded by the vertical lines through $l$ and $r$. Let $Q=\{l,a,b,c,r\}$. We consider three cases.

	\begin{figure}[htb]
	\centering
	\setlength{\tabcolsep}{0in}
	$\begin{tabular}{cccc}
		\multicolumn{1}{m{.25\columnwidth}}{\centering\vspace{0pt}\includegraphics[width=.23\columnwidth]{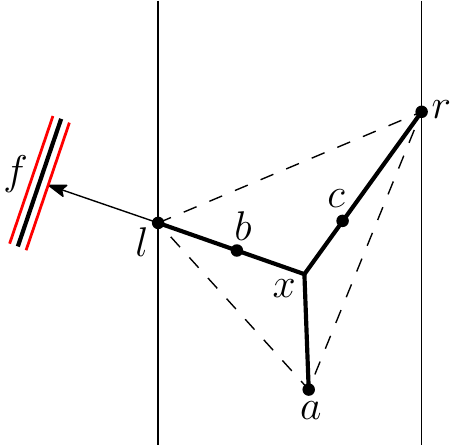}}
		&\multicolumn{1}{m{.25\columnwidth}}{\centering\vspace{0pt}\includegraphics[width=.22\columnwidth]{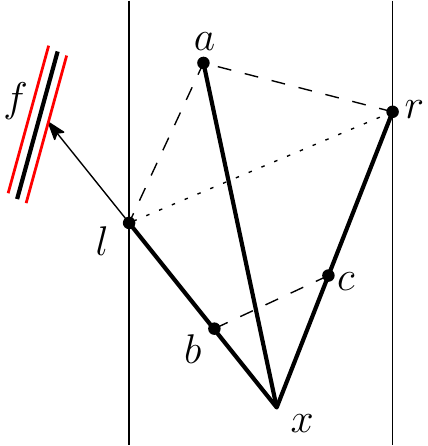}}
		&\multicolumn{1}{m{.25\columnwidth}}{\centering\vspace{0pt}\includegraphics[width=.22\columnwidth]{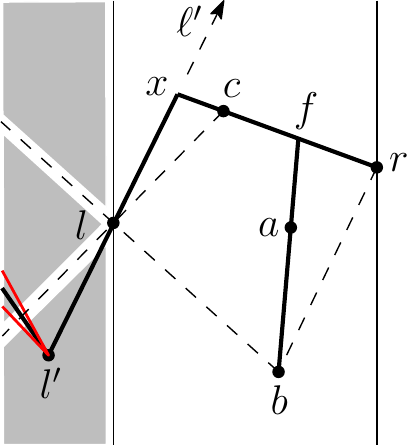}}
		&\multicolumn{1}{m{.25\columnwidth}}{\centering\vspace{0pt}\includegraphics[width=.22\columnwidth]{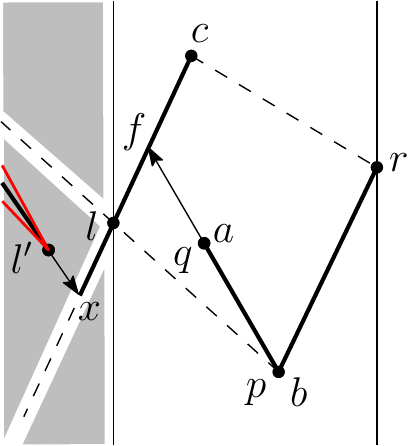}}
		\\
		\multicolumn{1}{m{.25\columnwidth}}{\centering\vspace{0pt}\includegraphics[width=.23\columnwidth]{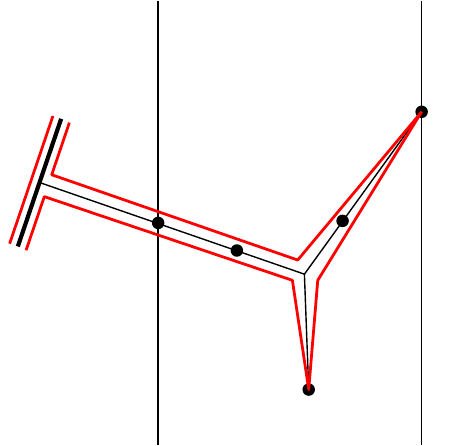}}
		&\multicolumn{1}{m{.25\columnwidth}}{\centering\vspace{0pt}\includegraphics[width=.22\columnwidth]{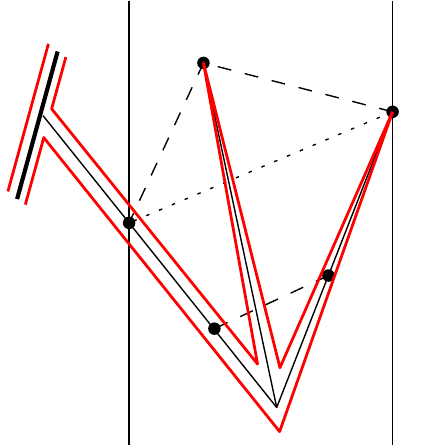}}
		&\multicolumn{1}{m{.25\columnwidth}}{\centering\vspace{0pt}\includegraphics[width=.22\columnwidth]{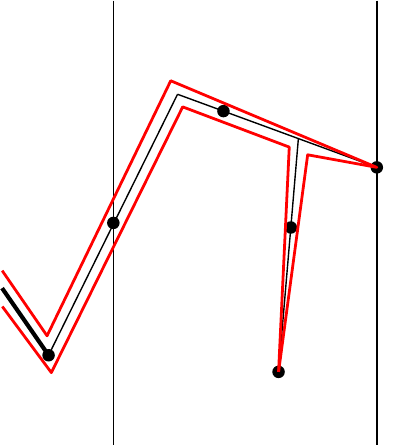}}
		&\multicolumn{1}{m{.25\columnwidth}}{\centering\vspace{0pt}\includegraphics[width=.22\columnwidth]{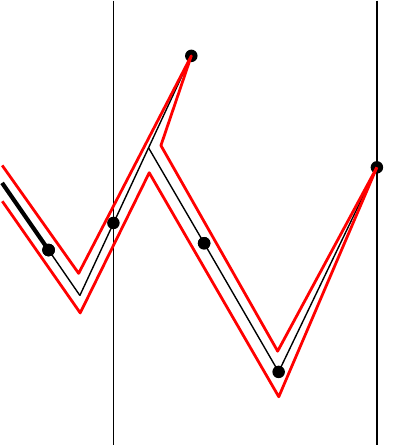}}
		\\
		(a)   &(b) &(c)&(d)
	\end{tabular}$
	\caption{Illustration of the proof of Theorem~\ref{tree-polygon-thr}. (a) $\conv{Q}$ has three vertices. (b) $\conv{Q}$ has five vertices. (c)-(d) $\conv{Q}$ has four vertices.}
	\label{polygon-fig}
\end{figure}

	\begin{itemize}
		\item {\em $\conv{Q}$ has three vertices.} Let $l$, $r$, $a$ be the vertices of $\conv{Q}$. Let $l,b,c,r$ be the four points that form a convex quadrilateral, and appear in this order along its boundary. Let $x$ be the intersection point of $\ell(l,b)$ and $\ell(r,c)$ which lies in $\bigtriangleup lra$. Let $f$ be the first intersection point of $T$ with the ray emanating from $x$ and passing through $l$, as in Figure~\ref{polygon-fig}(a). We cover the points by adding three segments $xf$, $xr$, and $xa$ to $T$. This generates only one new fork which is $f$ and it has multiplicity 1. Thus the invariant holds for the next iteration. 
		{\em \item $\conv{Q}$ has five vertices.} In this case $Q$ has a $4$-cap or a $4$-cup with endpoints $l$ and $r$. Let $l,b,c,r$ be such a cap or cup, as in Figure~\ref{polygon-fig}(b). We define $x$ and $f$ as in the previous case, and then cover the points by adding three segments $xf$, $xr$, and $xa$ to $T$. This generates one new fork which is $f$ and it has multiplicity 1. Thus the invariant holds for the next iteration. 
		{\em \item $\conv{Q}$ has four vertices.} After a suitable relabeling assume that $b$ and $c$ are on $\conv{Q}$. If both $b$ and $c$ lie above or below $\ell(l,r)$ then $l,b,c,r$ form a $4$-cap or a $4$-cup, in which case we cover the points as in the previous case. Thus we may assume $b$ lies below $\ell(l,r)$ and $c$ lies above $\ell(l,r)$. The lines $\ell(b,l)$ and $\ell(c,l)$ partition the halfplane to the left of $l$ into three regions (shaded regions in Figures~\ref{polygon-fig}(c) and \ref{polygon-fig}(d)). We distinguish two cases depending on the containment of $l'$ in these regions (recall that $l'$ is the rightmost scanned point in the previous iteration).
		
		\begin{itemize}
			\item {\em $l'$ lies above $\ell(b,l)$ or below $\ell(c,l)$.} Up to symmetry assume that $l'$ lies below $\ell(c,l)$; see Figure~\ref{polygon-fig}(c). Let $\ell'$ be the ray emanating from $l'$ and passing through $l$. Then $a$, $b$, $c$, and $r$ lie below $\ell'$. Let $x$ be the intersection point of $\ell(r,c)$ with $\ell'$ which lies in $H$. Consider the ray emanating from $b$ and passing through $a$. This ray hits $lx$ or $rx$ at a point which we denote by $f$. 
		We cover the points by adding three segments $xr$, $xl'$, and $bf$ to $T$. This generates only one new fork which is $f$ and it has multiplicity 1. Thus the invariant holds for the next iteration.

			\item {\em $l'$ lies below $\ell(b,l)$ and above $\ell(c,l)$.} This case is depicted in Figure~\ref{polygon-fig}(d). By the invariant, $l'$ has degree one in $T$, i.e. it is incident to exactly one segment in $T$. We extend this segment to the right until it intersects $\ell(b,l)$ or $\ell(c,l)$ for the first time. Up to symmetry we assume that the extension hits $\ell(c,l)$; we denote the intersection point by $x$. Among $a$ and $b$, let $p$ denote the one with a larger $x$-coordinate and $q$ denote the other one. Then the ray emanating from $p$ and passing through $q$ intersects a segment of $T$ or the segment $xc$. Consider the first such segment, and let $f$ denote the intersection point. We cover the points by adding three segments $xc$, $rp$, and $pf$ to $T$. (The segment incident to $l'$ was just extended to $x$.) Only one new fork is generated which is $f$ and it has multiplicity 1. Thus the invariant holds for the next iteration. 
		\end{itemize}
	\end{itemize}

	This is the end of an intermediate iteration. The noncrossing property of the resulting tree follows from our construction.

	For the first iteration of the algorithm we scan the two leftmost points of $K$, namely $p_1$ and $p_2$ where $p_1$ is to the left of $p_2$. We add to $T$ (which is initially empty) a long vertical segment $s_1$ that goes through $p_1$. Then we add to $T$ a second segment $s_2$ that connects $p_2$ to the midpoint of $s_1$. This midpoint is a fork with multiplicity $1$. The point $p_2$ will play the role of $l'$ for the second iteration (which is the first intermediate iteration). Notice that the invariant holds at this point as we have covered the $m=2$ scanned points by $2$ segments and $1$ fork.
	In the last iteration of the algorithm we are left with $w=(k-2)-5\lfloor(k-2)/5\rfloor\leqslant 4$ points. We connect these $w$ points by an $x$-monotone path with $w$ segments to the rightmost scanned point. Therefore, the final tree has at most $3\lfloor(k-2)/5\rfloor+2+w\leqslant\lceil3k/5\rceil+2$ segments with at most $\lfloor (k-2)/5\rfloor+1\leqslant\lceil k/5\rceil$ forks of multiplicity $1$.

Each iteration takes constant time. Therefore, after rotating and sorting the points in $O(k\log k)$ time, the rest of the algorithm takes linear time.
\end{proof}

\section{Concluding Remarks}
A natural open problem is to improve the presented upper bounds or the known lower bounds for $\pi(n)$, $\tau(n)$, and $\rho(k)$.
Here are some directions for further improvements on $\pi(n)$ and $\tau(n)$:
\begin{itemize}
	\item For the proof of Lemma~\ref{27-lemma} we used a 5-cap or a 5-cup which forced us to scan 21 points in each iteration (due to the result of Erd\H{o}s and Szekeres). If one could manage to use a 4-cap or a 4-cup instead, then it could improve the upper bound for $\pi(n)$ further. 
	\item Our iterative algorithm in the proof of Theorem~\ref{tree-thr}, covers $4$ points by $3$ edges in all cases except in the last case (where $\conv{Q}$ has four vertices and $a$ lies in $\bigtriangleup rbc$) for which it covers $5$ points by $4$ edges. The upper bound $4n/5$ for $\tau(n)$ comes from this case. If one could argue that this case won't happen often (for example by showing that it won't happen in three consecutive iterations or by choosing a different ordering for points), then it would lead to a slightly improved upper bound for $\tau(n)$. 
\end{itemize}

\bibliographystyle{abbrv}
\bibliography{Covering-Paths}

\begin{thebibliography}{10}

\bibitem{Aggarwal1999}
A.~Aggarwal, D.~Coppersmith, S.~Khanna, R.~Motwani, and B.~Schieber.
\newblock The angular-metric traveling salesman problem.
\newblock {\em {SIAM} Journal on Computing}, 29(3):697--711, 1999.
\newblock Also in {\em SODA'97}.

\bibitem{Aichholzer2010}
O.~Aichholzer, S.~Cabello, R.~F. Monroy, D.~Flores{-}Pe{\~{n}}aloza, T.~Hackl,
  C.~Huemer, F.~Hurtado, and D.~R. Wood.
\newblock Edge-removal and non-crossing configurations in geometric graphs.
\newblock {\em Discrete Mathematics \& Theoretical Computer Science},
  12(1):75--86, 2010.

\bibitem{Aichholzer2020}
O.~Aichholzer, F.~Duque, R.~F. Monroy, O.~E. Garc{\'{\i}}a{-}Quintero, and
  C.~Hidalgo{-}Toscano.
\newblock An ongoing project to improve the rectilinear and the pseudolinear
  crossing constants.
\newblock {\em Journal of Graph Algorithms and Applications}, 24(3):421--432,
  2020.

\bibitem{Arkin1999}
E.~M. Arkin, Y.~Chiang, J.~S.~B. Mitchell, S.~Skiena, and T.~Yang.
\newblock On the maximum scatter traveling salesperson problem.
\newblock {\em {SIAM} Journal on Computing}, 29(2):515--544, 1999.
\newblock Also in {\em SODA'97}.

\bibitem{Arkin2003b}
E.~M. Arkin, J.~S.~B. Mitchell, and C.~D. Piatko.
\newblock Minimum-link watchman tours.
\newblock {\em Information Processing Letters}, 86(4):203--207, 2003.

\bibitem{Arora1998}
S.~Arora.
\newblock Polynomial time approximation schemes for {E}uclidean traveling
  salesman and other geometric problems.
\newblock {\em Journal of the {ACM}}, 45(5):753--782, 1998.

\bibitem{Bereg2009}
S.~Bereg, P.~Bose, A.~Dumitrescu, F.~Hurtado, and P.~Valtr.
\newblock Traversing a set of points with a minimum number of turns.
\newblock {\em Discrete \& Computational Geomometry}, 41(4):513--532, 2009.
\newblock Also in {\em SoCG'07}.

\bibitem{Biniaz2022}
A.~Biniaz.
\newblock Euclidean bottleneck bounded-degree spanning tree ratios.
\newblock {\em Discrete \& Computational Geometry}, 67(1):311--327, 2022.
\newblock Also in {\em SODA'20}.

\bibitem{Cerny2007}
J.~Cern{\'{y}}, Z.~Dvor{\'{a}}k, V.~Jel{\'{\i}}nek, and J.~K{\'{a}}ra.
\newblock Noncrossing {H}amiltonian paths in geometric graphs.
\newblock {\em Discrete Applied Mathematics}, 155(9):1096--1105, 2007.

\bibitem{Chan2004}
T.~M. Chan.
\newblock {E}uclidean bounded-degree spanning tree ratios.
\newblock {\em Discrete {\&} Computational Geometry}, 32(2):177--194, 2004.
\newblock Also in {\em SoCG} 2003.

\bibitem{Chen2020}
J.~Chen, Q.~Huang, I.~Kanj, and G.~Xia.
\newblock Near-optimal algorithms for point-line covering problems.
\newblock {\em CoRR}, abs/2012.02363, 2020.

\bibitem{Collins2004}
M.~J. Collins.
\newblock Covering a set of points with a minimum number of turns.
\newblock {\em International Journal of Computational Geometry \&
  Applications}, 14(1-2):105--114, 2004.

\bibitem{Czabarka2008}
{\'{E}}.~Czabarka, O.~S{\'{y}}kora, L.~A. Sz{\'{e}}kely, and I.~Vrto.
\newblock Biplanar crossing numbers. {II.} {C}omparing crossing numbers and
  biplanar crossing numbers using the probabilistic method.
\newblock {\em Random Structures \& Algorithms}, 33(4):480--496, 2008.

\bibitem{Demaine2011}
E.~D. Demaine and J.~O'Rourke.
\newblock Open problems from {CCCG} 2010.
\newblock In {\em Proceedings of the 22nd Canadian Conference on Computational
  Geometry}, 2011.

\bibitem{Dumitrescu2014}
A.~Dumitrescu, D.~Gerbner, B.~Keszegh, and C.~D. T{\'{o}}th.
\newblock Covering paths for planar point sets.
\newblock {\em Discrete \& Computational Geometry}, 51(2):462--484, 2014.

\bibitem{Erdos1935}
P.~Erd\H{o}s and G.~Szekeres.
\newblock A combinatorial problem in geometry.
\newblock {\em Compositio Mathematica}, 2:463--470, 1935.

\bibitem{Fekete2000}
S.~P. Fekete and H.~Meijer.
\newblock On minimum stars and maximum matchings.
\newblock {\em Discrete {\&} Computational Geometry}, 23(3):389--407, 2000.
\newblock Also in {\em SoCG} 1999.

\bibitem{Fekete1997}
S.~P. Fekete and G.~J. Woeginger.
\newblock Angle-restricted tours in the plane.
\newblock {\em Computational Geometry: Theory and Applications}, 8:195--218,
  1997.

\bibitem{Flores-Penaloza21}
D.~Flores{-}Pe{\~{n}}aloza, M.~Kano, L.~Mart{\'{\i}}nez{-}Sandoval, D.~Orden,
  J.~Tejel, C.~D. T{\'{o}}th, J.~Urrutia, and B.~Vogtenhuber.
\newblock Rainbow polygons for colored point sets in the plane.
\newblock {\em Discrete Mathematics}, 344(7):112406, 2021.

\bibitem{Fulek2013}
R.~Fulek, B.~Keszegh, F.~Mori\'{c}, and I.~Uljarevi\'{c}.
\newblock On polygons excluding point sets.
\newblock {\em Graphs and Combinatorics}, 29(6):1741--1753, 2013.

\bibitem{Grantson2006}
M.~Grantson and C.~Levcopoulos.
\newblock Covering a set of points with a minimum number of lines.
\newblock In {\em Proceedings of the 6th International Conference on Algorithms
  and Complexity $($CIAC$)$}, pages 6--17, 2006.

\bibitem{Harary1963}
F.~Harary and A.~Hill.
\newblock On the number of crossings in a complete graph.
\newblock {\em Proceedings of the Edinburgh Mathematical Society}, 13:333--338,
  1963.

\bibitem{Ivanov2012}
A.~O. Ivanov and A.~A. Tuzhilin.
\newblock The {S}teiner ratio {G}ilbert-{P}ollak conjecture is still open:
  Clarification statement.
\newblock {\em Algorithmica}, 62(1-2):630--632, 2012.

\bibitem{Jiang2015}
M.~Jiang.
\newblock On covering points with minimum turns.
\newblock {\em International Journal of Computational Geometry \&
  Applications}, 25(1):1--10, 2015.

\bibitem{Keszegh2013}
B.~Keszegh.
\newblock Covering paths and trees for planar grids.
\newblock {\em Geombinatorics Quarterly}, 24, 2014.

\bibitem{Langerman2005}
S.~Langerman and P.~Morin.
\newblock Covering things with things.
\newblock {\em Discrete \& Computational Geometry}, 33(4):717--729, 2005.
\newblock Also in {\em ESA'02}.

\bibitem{Loyd1914}
S.~Loyd.
\newblock {\em Cyclopedia of 5000 Puzzles, Tricks \& Conundrums}.
\newblock The Lamb Publishing Company, 1914.

\bibitem{Pach2019}
J.~Pach, N.~Rubin, and G.~Tardos.
\newblock Planar point sets determine many pairwise crossing segments.
\newblock {\em Advances in Mathematics}, 386:107779, 2021.
\newblock Also in {\em STOC'19}.

\bibitem{Papadimitriou1977}
C.~H. Papadimitriou.
\newblock The {E}uclidean traveling salesman problem is {NP}-complete.
\newblock {\em Theoretical Computer Science}, 4(3):237--244, 1977.

\bibitem{Stein2001}
C.~Stein and D.~P. Wagner.
\newblock Approximation algorithms for the minimum bends traveling salesman
  problem.
\newblock In {\em Proceedings of the 8th International Conference on Integer
  Programming and Combinatorial Optimization {$($IPCO$)$}}, pages 406--422,
  2001.

\end{thebibliography}
\end{document}